\documentclass[a4paper,twoside]{llncs}

\let\fns\footnotesize

\usepackage{times}
\usepackage{url}
\usepackage{graphicx}
\usepackage{version}
\usepackage{alltt}
\usepackage{multirow}
\usepackage{bigdelim}       
\usepackage{amstext}
\usepackage{amsmath}
\usepackage{amssymb}
\usepackage{color}
\usepackage{textcomp}
\usepackage{rotating}
\usepackage{comment}

\usepackage{subfigure}
\usepackage[small]{caption}
\usepackage{floatflt}
\usepackage{enumerate}        

\usepackage{mathrsfs}         

\usepackage{bussproofs}       
\usepackage{dsfont}           
\usepackage{marginnote}

\usepackage{changepage}

\usepackage{longtable}
\usepackage{colortbl}
\usepackage{stmaryrd} 
\usepackage{textcomp} 

\usepackage{graphicx}
\usepackage{type1cm}
\usepackage{eso-pic}
\usepackage{color}
\makeatletter
\def\thewatermark{\begin{tabular}{c}THIS SECTION\\NOT FOR PUBLICATION\end{tabular}}
\def\watermarkangle{45}
\def\watermarklightness{0.9}
\def\watermarkfontsize{2cm}
\def\watermark{
    \AddToShipoutPicture{%
            \setlength{\@tempdimb}{.5\paperwidth}%
            \setlength{\@tempdimc}{.5\paperheight}%
            \setlength{\unitlength}{1pt}%
            \put(\strip@pt\@tempdimb,\strip@pt\@tempdimc){%
        \makebox(-20,100){\hss\rotatebox{\watermarkangle}{\textcolor[gray]{\watermarklightness}%
                 {\fontsize{\watermarkfontsize}{\watermarkfontsize}\selectfont{\thewatermark}}}\hss}%
        }%
    }%
}

\makeatother

\def\fract#1#2[#3]{
\frac{
\mbox{\fns$#1$}
}{
\mbox{\fns$#2$}
}
\mbox{\fns[$#3$]}
}
\def\pre[#1]#2{[#1]\mathop{{#2}}}
\def\EOP{\fbox{\vbox to 1ex{}~}}
\def\to{{\rightarrow\kern0.5pt}}

\def\E{{\cal E}}
\def\R{{\bf R}}
\def\G{{\bf G}}
\def\B{{\bf B}}
\def\N{{\bf N}}
\def\e{{\bf E}}
\def\code#1{{\bf#1}}
\newenvironment{example*}{%
  \renewcommand{\theexample}{\it Example~\arabic{example}.}
  \let\oldlabelcommand=\label
  \def\label{\customlabel{\theexample}}
  \let\oldexample=\example
  \let\oldendexample=\endexample
  \renewenvironment{example}{%
    \item[\hskip\labelsep\bfseries\refstepcounter{example}\theexample]%
  }{%
  }%
  \begin{trivlist}
  \item[\hskip\labelsep{\bfseries Examples}]%
}{%
  \end{trivlist}
  \let\label=\oldlabelcommand
  \let\example=\oldexample
  \let\endexample=\oldendexample
}
\newenvironment{definition*}{%
  \renewcommand{\thedefinition}{\Alph{section}D\arabic{definition}}
  \let\oldlabelcommand=\label
  \def\label{\customlabel{\thedefinition}}
  \let\olddefinition=\definition
  \let\oldenddefinition=\enddefinition
  \renewenvironment{definition}{%
    \item[\hskip\labelsep\bfseries\refstepcounter{definition}\thedefinition]%
  }{%
  }%
  \begin{trivlist}
  \item[\hskip\labelsep{\bfseries Definitions}]%
}{%
  \end{trivlist}
  \let\label=\oldlabelcommand
  \let\definition=\olddefinition
  \let\enddefinition=\oldenddefinition
}

\renewcommand{\implies}{~\mathop{\Rightarrow}~}
\renewcommand{\iff}{~\mathop{\Leftrightarrow}~}

\def\p#1{\mathrel{\ooalign{\hfil$\mapstochar\mkern 5mu$\hfil\cr$#1$}}}
\def \pfun{\p\rightarrow}

\makeatletter
\newenvironment{textbox}[1][]{%
  \def\@captype{box}%
    \par\nobreak\vspace{-2ex}\begin{center}\nobreak\begin{framebox}%
}{%
  \end{framebox}\par\nobreak\end{center}\vspace{-0.5ex}%
}
\newcounter{textbox}

\def\newdef#1#2{%
    \expandafter\@ifdefinable\csname #1\endcsname
        {\@definecounter{#1}%
         \expandafter\xdef\csname the#1\endcsname{\@thmcounter{#1}}%
         \global\@namedef{#1}{\@defthm{#1}{#2}}%
         \global\@namedef{end#1}{\@endtheorem}%
    }%
}
\def\@defthm#1#2{%
    \refstepcounter{#1}%
    \@ifnextchar[{\@ydefthm{#1}{#2}}{\@xdefthm{#1}{#2}}%
}
\def\@xdefthm#1#2{%
    \@begindef{#2}{\csname the#1\endcsname}%
    \ignorespaces
}
\def\@ydefthm#1#2[#3]{%
    \trivlist
    \item[%
        \hskip 10\p@
        \hskip \labelsep
        {\it #2%
         \saveb@x\@tempboxa{#3}
         \ifdim \wd\@tempboxa>\z@
            \ \box\@tempboxa
         \fi.%
        }]%
    \ignorespaces
}
\def\@begindef#1#2{%
    \trivlist
    \item[%
        \hskip 10\p@
        \hskip \labelsep
        {\it #1\ \rm #2.}%
    ]%
}
\makeatother

\newdef{fact}{Fact}

\renewcommand{\appendix}{\par
\section*{APPENDIX -- {\em NOT FOR PUBLICATION\/}}
\setcounter{section}{0}
 \setcounter{subsection}{0}
  \def\thesection{\Alph{section}} }

\excludecomment{anonymous}\includecomment{standard}
\excludecomment{withappendix}\includecomment{withoutappendix}

\begin{anonymous}

\def\citeBB#1{\ifnum#1=11 \expandafter\mycite{BB11A}\else\ifnum#1=94 \expandafter\mycite{BB94A}\else\ifnum#1=13 \expandafter\mycite{BB13A}\fi\fi\fi}

\end{anonymous}
\begin{standard}

\def\citeBB#1{\ifnum#1=11 \expandafter\mycite{BB11}\else\ifnum#1=94 \expandafter\mycite{BB94}\else\ifnum#1=13 \expandafter\mycite{BB13}\fi\fi\fi}

\def\andciteBBtwo#1,{BB#1,\andciteBBthree}
\def\andciteBBthree#1,{BB#1}

\end{standard}

\def\mycite#1{\cite{#1}}

\def\paragraph#1{{\bf#1}}

\begin{standard}
\def\authorA{Peter~T.~Breuer}
\def\authorAemail{ptb@cs.bham.ac.uk}

\def\authorAaddr{Department of Computer Science, University of Birmingham, UK}
\def\authorB{Simon~J.~Pickin}
\def\authorBemail{spickin@ucm.es}

\def\authorBaddr{Facultad de Inform\'atica, Universidad Complutense de Madrid}
\end{standard}

\begin{anonymous}
\def\authorA{Author~1}
\def\authorAemail{Author 1 email}

\def\authorAaddr{Author 1 address}
\def\authorB{Author~2}
\def\authorBemail{Author 2 email}

\def\authorBaddr{Author 2 address}
\end{anonymous}

\def\nedots{{.^{.^.}\kern-5pt}}

\begin{document}

\edef\marginnotetextwidth{\the\textwidth}

\title{\kern-30pt Soundness and Completeness of the NRB Verification
Logic\kern-30pt}

\author{
\authorA\inst{1}
\and
\authorB\inst{2}
}
\institute{
  \authorAaddr\\
  \email{\authorAemail}
\and
  \authorBaddr\\
  \email{\authorBemail}
}

\maketitle

\begin{abstract}
This short paper gives a model for and a proof of completeness of
the NRB verification logic for deterministic imperative programs, the
logic having been used in the past as the basis for automated
semantic checks of large, fast-changing, open source C code archives,
such as that of the Linux kernel source. The model is a coloured
state transitions model that approximates from above the set of
transitions possible for a program. Correspondingly, the logic
catches all traces that may trigger a particular defect at a given
point in the program, but may also flag false positives.
\end{abstract}


\pagestyle{plain}

\section{Introduction}
\label{sec:Introduction}

NRB program logic was first introduced in 2004 \cite{RST:2004} as
the theory supporting an automated semantic analysis suite
\cite{iCCS2006} targeting the C code of the Linux kernel.  The analyses
performed with this kind of program logic and automatic tools are
typically much more approximate than that provided by more interactive
or heavyweight techniques such as theorem-proving and model-checking
\cite{Clarke1}, respectively, but the NRB combination has proved
capable of rapidly scanning millions of lines of C code and detecting
deadlocks scattered at one per million lines of code \cite{SEW30}.  A
rough synopsis of the characteristics of the logic or an approach using
the logic is that it is precise in terms of accurately following the
often complex flow of control and sequence of events in an imperative
language, but not very accurate at following data values.  That is fine
for a target language like C \cite{C89,C99}, where static analysis
cannot reasonably hope to follow all data values accurately because of
the profligate use of indirection through pointers in a typical program
(a pointer may access any part of memory, in principle, hence writing
through a pointer might `magically' change any value) and the NRB
logic was designed to work around that problem by focussing instead on
information derived from sequences of events.

NRB is a logic with modal operators.  The modalities do not 
denote a full range of actions as in Dynamic Logic ~\cite{Harel:2000},
but rather only the very particular action of
the final exit from a code fragment being via a
\code{return}, \code{break}, or \code{goto}.
The logic is also configurable in detail to support the
code abstractions  that are of interest in different analyses;
detecting the freeing of a record in memory while it may
still be referenced requires an abstraction that counts the possible
reference holders, for example, not the value currently in the second 
field from the right.  The technique became known as `symbolic
approximation' \cite{SA,ISOLA} because of the foundation in symbolic logic 
and because the analysis is guaranteed to be on the alarmist side
(`approximate from above'); the analysis does not miss bugs in 
code, but does report false positives.
In spite of a few years' pedigree behind it now, a
foundational semantics for the logic has only just been published
\cite{SCP} (as an Appendix to the main text), and this article
aims to provide a yet simpler semantics for the logic and also a
completeness result, with the aim of consolidating the technique's
bona fides.

Interestingly, the formal guarantee (`never miss, over-report')
provided by NRB and the symbolic approximation technique is said not to
be desirable in the commercial context by the very practical authors of
the Coverity analysis tool \cite{Coverity,Bessey:2010}, which also has
been used for static analysis of the Linux kernel and many very large C
code projects.  Allegedly, in the commercial arena, understandability
of reports is crucial, not the guarantee that no bugs will be missed.
The Coverity authors say that commercial clients tend to dismiss any
reports that they do not understand, turning a deaf ear to explanations.
However, the reports produced by our tools have always been filtered
before presentation, so only the alarms that cannot
be dismissed as false positives are seen.

The layout of this paper is as follows. In Section~\ref{sec:A1} a model
of programs as sets of `coloured' transitions between states is
introduced, and the constructs of a generic imperative language are
expressed in those terms.  It is shown that the constructs obey certain
algebraic laws, which soundly implement the established deduction rules
of NRB logic. Section~\ref{sec:A2} shows that the logic is complete for
deterministic programs, in that anything that is true in the model
introduced in Section~\ref{sec:A1} can be proved using the formal rules 
of the NRB logic.

Since the model contains at least as many state transitions as occur in
reality, `soundness' of the NRB logic means that it may construct
false alarms for when a particular condition may be breached at some
particular point in a program, but that it may not miss any real alarms.
`Completeness' means that the logic flags no more false alarms than
are already to be predicted from the model, so if the model says that
there ought to be no alarms at all (which means that there really
are no alarms), then the logic can prove that.  Thus, reasoning
symbolically is not in principle an approximation here; it is not
necessary to laboriously construct and examine the complete graph of
modelled state transitions in order to be able to give a program a `clean
bill of health' with reference to some potential defect, because the
logic can always do the job as well.

\section{Semantic Model}
\label{sec:A1}
\begin{table}[t]
\caption{NRB deduction rules for triples of assertions and
programs. Unless explicitly noted, assumptions
${\G}_l p_l$ at left are passed down unaltered from top to bottom of each
rule.
We let $\E{}_1$ stand for
any of ${\R}$, ${\B}$, ${\G}_l$, ${\e}_k$; $\E_2$ any of ${\R}$, ${\G}_l$, ${\e}_k$;
$\E_3$ any of ${\R}$.  ${\G}_{l'}$ for $l'\ne l$, ${\e}_k$;
$\E_4$ any of ${\R}$.  ${\G}_l$, ${\e}_{k'}$ for $k'\ne k$; $[h]$ the body of
the subroutine named $h$.
}
\label{tab:rules}
\[
\begin{array}{c}
\frac{
\triangleright~\{p\}\, P\, \{{\N}q\lor \E_{1}x\}
\quad
\triangleright~\{q\}\, Q\, \{{\N}r\lor \E_{1}x\}
}{
\triangleright~\{p\}\, P\,{;}\,Q\,\, \{{\N}r\lor \E_{1}x\}
}\mbox{\footnotesize[seq]}
\qquad
\frac{
\triangleright~\{p\}\, P\, \{\B q\lor {\N}p\lor \E_{2}x\}
}{
\triangleright~\{p\}\, \code{do} \,\,P\, \{{\N}q\lor \E_{2}x\}
}\mbox{\footnotesize[do]}
\\[2ex]
\frac{
}{
\triangleright~\{p\}\, \code{skip}\, \{{\N}\,p\}
}\mbox{\footnotesize[skp]}
\qquad
\frac{
}{
\triangleright~\{p\}\, \code{return}\, \{{\R}\,p\}
}\mbox{\footnotesize[ret]}
\\[2ex]
\frac{
}{
\triangleright~\{p\}\, \code{break}\,\,\, \{{\B}\,p\}
}\mbox{\footnotesize[brk]}
\quad
\mbox{\footnotesize[$p\to p_l$]}
\frac{
}{
{\G}_l\,p_l\,\triangleright~\{p\}\, \code{goto}\,\,l\, \{{\G}_l\,p\}
}\mbox{\footnotesize[go]}
\\[2ex]
\frac{
}{
\triangleright~\{p\}\, \code{throw}\,\,k\, \{{\e}_k\,p\}
}\mbox{\footnotesize[throw]}
\quad
\frac{
}{
\triangleright~ \{q[e/x] \}\,\, x{=}e\,\, \{{\N}q\}
}\mbox{\footnotesize[let]}
\\[2ex]
\frac{
\triangleright~ \{q\land p\}\, P\, \{r\}
}{
\triangleright~\{p\}\, q\,{\to}P\,\, \{r\}
}\mbox{\footnotesize[grd]}
\quad
\frac{
\triangleright~ \{p \}\, P\,\, \{q\}
\quad
\triangleright~ \{p \}\, Q\,\, \{q\}
}{
\triangleright~ \{p \}\, P\,{\shortmid}\,Q\,\, \{q\}
}\mbox{\footnotesize[dsj]}
\\[2ex]
\mbox{\footnotesize$[\N p_l\to q]$}
\frac{
{\G}_l\,p_l\,\,\triangleright~ \{p \}~ P~ \{q\}
}{
{\G}_l\,p_l\,\,\triangleright~ \{p \}~ P:l~ \{q\}
}\mbox{\footnotesize[frm]}
\qquad
\frac{
{\G}_l\,p_l\,\,\triangleright~ \{p \}~ P~ \{{\G}_l p_l \lor {\N}q \lor \E_{3}x\}
}{
\triangleright~ \{p \}~ \code{label}~l.P~ \{{\N}q\lor \E_{3}x\}
}\mbox{\footnotesize[lbl]}
\\[2ex]
\frac{
\triangleright~ \{p \}~ [h]~ \{\R r \lor {\e}_kx_k\}
}{
{\G}_l p_l\,\triangleright~ \{p \}~ \code{call}~h~ \{{\N}r\lor {\e}_kx_k\}
}\mbox{\footnotesize[sub]}
\quad
\frac{
\triangleright~ \{p \}~ P~ \{{\N}r\lor {\e}_kq\lor\E_4 x\}
\quad
\triangleright~ \{q \}~ Q~ \{{\N}r\lor {\e}_k x_k\lor\E_4 x\}
}{
\triangleright~ \{p \}~ \code{try}~P~\code{catch}(k)~Q~\{{\N}r\lor {\e}_k x_k\lor\E_4x \}
}\mbox{\footnotesize[try]}
\\[2ex]
\frac{
\triangleright~\{p_i\}~P~\{q\}
}{
\triangleright~\{ {\lor}\kern-3pt{\lor} p_i\}~P~\{ q\}
}
\qquad
\frac{
\triangleright~\{p\}~P~\{q_i\}
}{
\triangleright~\{ p\}~P~\{ {\land}\kern-3pt{\land} q_i\}
}
\qquad
\frac{
{\G}_l\,p_{li}\,\triangleright~\{p\}~P~\{q\}
}{
{\lor}\kern-3pt{\lor} {\G}_l\,p_{li}\,\triangleright~\{ p\}~P~\{ q\}
}
\\[2ex]
\mbox{\footnotesize[$p'\to p, q\to q', p_l'\to p_l|{\G}_lq'\to {\G}_lp'_l$]}
\frac{
{\G}_l\,p_l\,\triangleright~\{ p\}~P~\{ q\}
}{
{\G}_l\,p_l'\,\triangleright~\{ p'\}~P~\{ q'\}
}
\end{array}
\]
\label{tab:NRBG}
\end{table}
This section sets out a semantic model for the full NRBG(E) 
logic (`NRB' for short) shown in Table~\ref{tab:NRBG}.  The `NRBG' part
stands for `normal, return, break, goto', and the `E' part treats
exceptions (catch/throw in Java, setjmp/longjmp in C), aiming at a
complete treatment of classical imperative languages.  This semantics
simplifies a {\em trace model} presented in the Appendix to \cite{SCP},
substituting traces there for state transitions here.

A natural model of a program is as a relation of type
$\mathds{P}(S\times S)$, expressing possible changes in a state of type $S$
as a set of pairs of initial and final states.  We shall add
a {\em colour} to this picture. The `colour' shows if the program has
run {\em normally} through to the end (colour `$\N$') or has terminated
early via a \code{return} (colour `$\R$'), \code{break} (colour `$\B$'),
\code{goto} (colour `$\G_l$' for some label $l$) or an exception (colour
`$\e_k$' for some exception kind $k$).  The aim is to document
precisely the control flow in the program. In this picture, a
deterministic program may be modelled as a set of `coloured'
transitions of type
\[
   \mathds{P}(S\times \star\times S)
\]
where the colours $\star$ are a disjoint union
\[
   \star = \{\N\}
       \sqcup
       \{\R\}
       \sqcup
       \{\B\}
       \sqcup
       \{\G_l\,|\,l\in L\}
       \sqcup
       \{\e_k\,|\,k\in K\}
\]
and $L$ is the set of possible \code{goto} labels and $K$ the set of possible
exception kinds.

The programs we consider are in fact deterministic, but we will use the
general setting.  Where the relation is not
defined on some initial state $s$, we understand that the
initial state $s$ leads to the program getting hung up in an infinite
loop, instead of terminating. Relations representing deterministic programs
thus have a set of images for any given initial state that is either of
size zero (`hangs') or one (`terminates').
Only paths through the program that do not `hang' in an infinite loop
are of interest to us, and what the NRB logic will say about a program
at some point will be true only supposing control reaches that point,
which it may never do.
%


Programs are put together in sequence  with the
second program accepting as inputs only the states that the first program ends `normally'
with. Otherwise the state with which the first program exited
abnormally is the final outcome. That is,
\begin{align*}
   \llbracket P;Q\rrbracket  &= \{ s_0\mathop{\mapsto}\limits^\iota s_1
   \in \llbracket P\rrbracket ~|~ \iota\ne\N\}\\
                            &\cup \,\{s_0\mathop{\mapsto}\limits^\iota
                            s_2 \mid s_1\mathop{\mapsto}\limits^\iota
                            s_2 \in \llbracket Q\rrbracket ,~
                            s_0\mathop{\mapsto}\limits^\N s_1\in \llbracket P\rrbracket \}  
\end{align*}
This statement is not complete, however, because abnormal exits with a
\code{goto} from $P$ may still re-enter in $Q$ if
the \code{goto} label is in $Q$, and proceed.
We postpone consideration of this eventuality by predicating the
model with the sets of states $g_l$ {\em hypothesised} as
being fed in at the label $l$ in the code. The model of
$P$ and $Q$ with these sets as assumptions  produce outputs that
take account of these putative extra  inputs at label $l$:
\begin{align*}
   \llbracket P;Q\rrbracket_g  &= \{ s_0\mathop{\mapsto}\limits^\iota s_1
   \in \llbracket P\rrbracket_g ~|~ \iota\ne\N\}\\
                            &\cup \,\{s_0\mathop{\mapsto}\limits^\iota
                            s_2 \mid s_1\mathop{\mapsto}\limits^\iota
                            s_2 \in \llbracket Q\rrbracket_g ,~
                            s_0\mathop{\mapsto}\limits^\N s_1\in
                            \llbracket P\rrbracket_g \}  
\end{align*}
Later, we will tie things up by ensuring that the set of states
bound to early exits via a \code{goto}~$l$ in $P$ are exactly the sets
$g_l$ hypothesised here as entries at label $l$ in $Q$ (and vice versa).
The type of the {\em interpretation} expressed by the fancy square
brackets is
\[
\llbracket{-_1}\rrbracket_{-_2} : {\mathscr C}\to(L\pfun \mathds{P}S)\to 
\mathds{P}(S\times \star\times S)
\]
where $g$, the second argument/suffix, has the partial
function type $L\pfun \mathds{P}S$ and the first argument/bracket interior
has type $\mathscr C$, denoting a simple language of imperative statements whose
grammar is set out in Table~\ref{tab:BNF}.  The models of some of its
very basic statements as members of $\mathds{P}(S\times\star\times S)$
are shown in Table~\ref{tab:ex1to4} and we will discuss them and the
interpretations of other language constructs below.

\begin{table}[tb]
\subtable{}{
\fbox{
\begin{minipage}[t]{0.465\textwidth}
A \code{skip} statement is modelled as
\[
\llbracket \code{skip} \rrbracket_g  =
   \{ s\mathop{\mapsto}\limits^{\N} s\mid s\in S \}
\]
It makes the
transition from a state to the same state again, and ends `normally'.
\end{minipage}
}
}
\subtable{}{
\fbox{
\begin{minipage}[t]{0.465\textwidth}
A \code{return} statement has the model 
\[
\llbracket \code{return} \rrbracket_g =
   \{ s\mathop{\mapsto}\limits^{\R} s\mid s\in S \}
\]
It exits at once `via a return flow' after a single,
trivial transition.
\end{minipage}
}
}
\subtable{}{
\fbox{
\begin{minipage}[t]{0.465\textwidth}
The model of $\code{skip};\code{return}$ is
\[
\llbracket \code{skip};\code{return}\rrbracket_g =
   \{ s\mathop{\mapsto}\limits^{\R} s\mid s\in S \}
\]
which is the same as that of \code{return}. It is made up of
the compound of two trivial state transitions, 
$s\mathop{\mapsto}\limits^{\N} s$ from \code{skip} and
$s\mathop{\mapsto}\limits^{\R} s$  from \code{return},
the latter ending in a `return flow'.
\medskip
\medskip
\medskip
\medskip
\medskip
\end{minipage}
}
}
\subtable{}{
\fbox{
\begin{minipage}[t]{0.465\textwidth}
The 
$\code{return};\code{skip}$ compound is modelled as:
\[
\llbracket \code{return};\code{skip}\rrbracket_g =
   \{ s\mathop{\mapsto}\limits^{\R} s\mid s\in S \}
\]
It is made up of of just the $s\mathop{\mapsto}\limits^{\R} s$
transitions from \code{return}.
There is no transition that can be formed as the
composition of a transition from \code{return} followed by a transition 
from \code{skip}, because none of the first end `normally'.
\end{minipage}
}
}
\medskip
\caption{Models of simple statements.}
\label{tab:ex1to4}
\end{table}


A real imperative programming language such as C can be mapped onto
$\mathscr C$ -- in principle exactly, but in practice rather
approximately with respect to data values, as will be indicated below.
\begin{table}[t]
\caption{Grammar of the abstract imperative language $\mathscr C$,
where integer variables $x\in X$, term expressions $e \in \mathscr E$,
boolean expressions $b \in \mathscr B$, labels $l \in L$, exceptions $k\in K$,
statements $c \in \mathscr C$, integer constants $n \in {\mathds{Z}}$,
infix binary relations $r \in R$, subroutine names $h \in H$.
Note that  labels (the targets of \code{goto}s) are declared with
`\code{label}' and a label cannot be the first thing in a code sequence;
it must follow some statement. Instead of \code{if}, $\mathscr C$ has guarded
statements, and explicit nondeterminism, which, however, is only to be
used here in the deterministic construct $b\to P \shortmid \lnot b\to
Q$ for code fragments $P$, $Q$.
}
\label{tab:BNF}
\footnotesize
\vspace{-2ex}
\begin{align*}
{\mathscr C}~{:}{:}&
{\text{=}}
       ~\code{skip} 
~{\mid}~\code{return}
~{\mid}~\code{break}
~{\mid}~\code{goto}\,\,l
~{\mid}~c{;}c
~{\mid}~ x {=} e
~{\mid}
       ~b{\to} c
~{\mid}~c\,{\shortmid}\,c
~{\mid}~\code{do}~c
~{\mid}~c\,{:}\,l
~{\mid}~\code{label}\,\,l.c
~{\mid}~\code{call}\,\,h\\
&\mid 
      ~\code{try}~c~\code{catch}(k)~c
~{\mid}~\code{throw}\,\,k
\\
\mathscr{E}~{:}{:}&
{\text{=}}
  ~n
  \mid x
  \mid n*e
  \mid e + e
  \mid b\,?\,e:e
\\
\mathscr{B}~{:}{:}&
{\text{=}}
   ~\top \mid \bot \mid e~r~e
   \mid b \lor b
   \mid b \land b
   \mid \lnot b
   \mid \exists x. b
\\
R~{:}{:}&
{\text{=}}
   ~{<}
   \mid {>}
   \mid {\le}
   \mid {\ge}
   \mid {=}
   \mid {\ne}
\end{align*}
\vspace{-4ex}
\end{table}
A conventional $\code{if}(b)~P~\code{else}~Q$ statement in C is
written as the nondeterministic choice between two guarded statements
$b\to P\shortmid\lnot b\to Q$ in the abstract language $\mathscr C$;
the conventional $\code{while}(b)~P$ loop in C is expressed as
$\code{do}\{\lnot b\to\code{break}\shortmid b\to P\}$, using
the forever-loop of $\mathscr C$, etc. A sequence $P; l: Q$ in C with a
label $l$ in
the middle should strictly be expressed as $P : l; Q$ in
$\mathscr C$, but we regard
$P ; l : Q$  as syntactic sugar for that, so it is still
permissible to write $P ; l: Q$ in $\mathscr C$. As a very special syntactic
sweetener, we permit $l : Q$ too, even when there is no preceding
statement $P$, regarding it as an abbreviation for $\code{skip} : l; Q$.

Curly brackets may be used to group code statements for clarity in
$\mathscr C$,
and parentheses may be used to group expressions.  The variables are
globals and are not formally declared.  The terms of $\mathscr C$ are piecewise
linear integer forms in integer variables, so the boolean expressions
are piecewise comparisons between linear forms.

\begin{example}
A valid integer term is `$\rm 5x + 4y + 3$', and a boolean expression is
`$\rm 5x + 4y + 3 < z - 4 \land y \le x$'.

In consequence another valid
integer term, taking the value of the first on the range defined by
the second, and 0 otherwise, is 
`$\rm (5x + 4y + 3 < z - 4 \land y \le x)\,?\,5x + 4y + 3:0$'.
\end{example}

\noindent
The limited set of terms in $\mathscr C$ makes it practically
impossible to map standard imperative language 
assignments as simple as `$\rm x=x*y$' or `$\rm x= x\mid y$' (the
bitwise or) succinctly.  In principle, those could be expressed exactly
point by point using conditional expressions (with at most $2^{32}$
disjuncts), but it is usual to model  all those cases by means of
an abstraction away from the values taken to attributes
that can be represented more elegantly using piecewise
linear terms The abstraction may be to how many times the variable has
been read since last written, for example, which maps `$\rm x= x*y$' to
`$\rm x = x+1; y = y+1; x = 0$'.  

Formally, terms have a conventional evaluation as integers and
booleans that is shown (for completeness!) in Table~\ref{tab:ev}.
The reader may note 
the notation $s\,x$ for the evaluation of the variable named $x$ in state
$s$, giving its integer value as result. We say that state $s$ {\em
satisfies} boolean term $b\in\mathscr B$, written $s\models b$, whenever $\llbracket
b\rrbracket s$ holds.
%
%
\begin{table}[t]
\caption{The conventional evaluation of integer and boolean terms of 
$\mathscr C$, for variables $x\in X$, integer constants $\kappa\in{\mathds{Z}}$,
using $s\,x$ for the (integer) value  of the variable named $x$ in a
state $s$. The form $b[n/x]$ means `expression $b$ with integer $n$
substituted for all unbound occurrences of $x$'.
}
\label{tab:ev}
\footnotesize
\[
\begin{array}[t]{@{}r@{~}c@{~}l}
\llbracket-\rrbracket&:&{\mathscr E} \to S \to {\mathds{Z}}\\
\llbracket x \rrbracket s&=&s\,x\\
\llbracket \kappa \rrbracket s&=&\kappa\\
\llbracket \kappa*e \rrbracket s&=&\kappa*\llbracket e \rrbracket s\\
\llbracket e_1 + e_2 \rrbracket s&=&\llbracket e_1 \rrbracket s
                                  + \llbracket e_2 \rrbracket s\\
\llbracket b\,?\,e_1:e_2 \rrbracket s&=&
      \mbox{if} ~\llbracket b \rrbracket s~ \mbox{then}~ \llbracket e_1
      \rrbracket s~\mbox{else}~ \llbracket e_2 \rrbracket s
\end{array}
\quad
\begin{array}[t]{r@{~}c@{~}l@{}}
\llbracket-\rrbracket&:&{\mathscr B}\to S \to \mbox{\bf bool}\\
\llbracket \top \rrbracket s&=&\top\qquad
\llbracket \bot \rrbracket s =\bot\\
\llbracket e_1 < e_2 \rrbracket s &=&\llbracket e_1 \rrbracket s
                                  < \llbracket e_2 \rrbracket s\\
\llbracket b_1 \lor b_2 \rrbracket s &=&\llbracket b_1 \rrbracket s
                                  \lor \llbracket b_2 \rrbracket s\\
\llbracket b_1 \land b_2 \rrbracket s &=&\llbracket b_1 \rrbracket s
                                  \land \llbracket b_2 \rrbracket s\\
\llbracket \lnot b  \rrbracket s &=&\lnot (\llbracket b \rrbracket s)\\
\llbracket \exists x. b\rrbracket s &=& \exists n\in \mathds{Z}.
                        \llbracket b[n/x] \rrbracket s
\end{array}
\]
\end{table}

%

The \code{label} construct of $\mathscr C$ declares a label $l\in L$ that may 
subsequently be used as the target in \code{goto}s.  The component $P$
of the construct is the body of code in which the label is {\em in
scope}.  A label may not be mentioned except in the scope of its
declaration.  The same label may not be declared again in the scope of
the first declaration.  The semantics of labels and \code{goto}s 
will be further explained below.

The only way of exiting the $\mathscr C$ \code{do} loop construct normally is via
\code{break} in the body $P$ of the loop.  An abnormal exit other than
\code{break} from the body $P$ terminates the whole loop
abnormally.  Terminating the body $P$ normally evokes one
more turn round the loop. So conventional \code{while} and \code{for}
loops need to be mapped to a \code{do} loop with a guarded
\code{break} statement inside, at the head of the body.
The precise models for this and every construct of $\mathscr C$ as a set of
coloured transitions are enumerated in Table~\ref{tab:interpretation}.

\begin{table}[t]
\caption{Model of programs of language $\mathscr C$,
given as hypothesis the sets of states $g_l$
for $l\in L$ observable at
$\code{goto}~l$ statements. A recursive reference means
`the least set satisfying the condition'.
For $h\in H$, the subroutine named $h$ has code $[h]$.
The state $s$ altered by the assignment of $n$ to variable $x$ is
written $s[x\mapsto n]$.
}
\label{tab:interpretation}
\footnotesize
\[
\begin{array}[t]{@{}r@{~}l@{~}}
\llbracket-\rrbracket_g&:~\mathscr C \to \mathds{P}(S\times\star\times S)\notag\\[0.5ex]
\llbracket\code{skip}\rrbracket_g &=
  \{s_0\mathop{\mapsto}\limits^{\N}s_0\mid s_0\in S\}\\
\llbracket\code{return}\rrbracket_g s_0&= 
  \{s_0\mathop{\mapsto}\limits^{\R}s_0\mid s_0\in S\}\\
\llbracket\code{break}\rrbracket_g &= 
  \{s_0\mathop{\mapsto}\limits^{\B}s_0\mid s_0\in S\}\\
\llbracket\code{goto}~l\rrbracket_g &=
  \{s_0\mathop{\mapsto}\limits^{\G_l}s_0\mid s_0\in S\}\\
\llbracket\code{throw}~k\rrbracket_g &=
  \{s_0\mathop{\mapsto}\limits^{\e_k}s_0\mid s_0\in S\}\\
\llbracket P;Q\rrbracket_g &= 
  \{ 
    s_0\mathop{\mapsto}\limits^\iota s_1 \in \llbracket P\rrbracket_g \mid \iota \ne \N
    \}\\
  &\cup~
  \{ 
     s_0\mathop{\mapsto}\limits^\iota s_2 \mid
         s_1\mathop{\mapsto}\limits^\iota s_2\in\llbracket Q\rrbracket_g
         ,~ s_0\mathop{\mapsto}\limits^{\N} s_1\in \llbracket P\rrbracket_g
  \}
\\
\llbracket x=e\rrbracket_g s_0&=
  \{ s_0\mathop{\mapsto}\limits^{\N} s_0[x\mapsto\llbracket
  e\rrbracket s_0]\}
  \mid s_0\in S
  \}
\\
\llbracket p \to P\rrbracket_g &= 
  \{ s_0\mathop{\mapsto}\limits^\iota s_1 \in \llbracket P\rrbracket_g \mid \llbracket p\rrbracket s_0\}
\\
\llbracket P\shortmid Q\rrbracket_g&= \llbracket P\rrbracket_g \cup
\llbracket Q\rrbracket_g 
\\
\llbracket \code{do}~P\rrbracket_g &= 
  \{ s_0\mathop{\mapsto}\limits^{\N} s_1 \mid
  s_0\mathop{\mapsto}\limits^{\B} s_1 \in \llbracket P\rrbracket_g \}\\
  &\cup~
  \{ s_0\mathop{\mapsto}\limits^\iota s_1 \in \llbracket P\rrbracket_g  \mid~\iota\ne {\N},{\B} \}\\
  &\cup~
  \{ s_0\mathop{\mapsto}\limits^\iota s_2 \mid
     s_1\mathop{\mapsto}\limits^\iota s_2\in \llbracket \code{do}~P\rrbracket_g 
  ,~
     s_0\mathop{\mapsto}\limits^\iota s_1 \in \llbracket P\rrbracket_g 
  \}
\\
\llbracket P : l \rrbracket_g &= 
  \llbracket P\rrbracket_g \\
  &\cup~
  \{ s_0\mathop{\mapsto}\limits^{\N} s_1 \mid s_0\in S,~s_1 \in g_l \}
\\
\llbracket \code{label}~ l~ P\rrbracket_g &=
 \llbracket P \rrbracket_{g\cup\{l\mapsto g_l^*\}}
 - g_l^*
 \label{eq:label}
 \\
 &~\mbox{where}~g_l^*  = 
 \{
 s_1 \mid
 s_0\mathop{\mapsto}\limits^{\G_l} s_1\in
 \llbracket P \rrbracket_{g\cup\{l\mapsto g_l^*\}}
 \}
\\
\llbracket \code{call}~h\rrbracket_g &=
 \{
 s_0\mathop{\mapsto}\limits^{\N} s_1 \mid
 s_0\mathop{\mapsto}\limits^{\R} s_1 \in \llbracket [h]\rrbracket_{\{\,\}}
 \}
 \\
  &\cup~
  \{ s_0\mathop{\mapsto}\limits^{\e_k} s_1 \in \llbracket [h]\rrbracket_{\{\,\}} \mid k\in K
  \}
\\
\llbracket \code{try}~P~\code{catch}(k)~Q\,\rrbracket_g &=
 \{
 s_0\mathop{\mapsto}\limits^\iota s_1 \in\llbracket P\rrbracket_g \mid ~\iota\ne {\e}_k
 \}\\
 &\cup~
 \,\{
 s_0\mathop{\mapsto}\limits^\iota s_2
 \mid
 s_1\mathop{\mapsto}\limits^\iota s_2\in \llbracket Q\rrbracket_g 
 ,~
 s_0\mathop{\mapsto}\limits^{\e_k} s_1\in \llbracket P\rrbracket_g 
 \}
\end{array}
\]
\end{table}
Among the list of models in Table~\ref{tab:interpretation}, that
of \code{label} declarations in particular requires 
explanation because labels are more explicitly controlled in $\mathscr
C$ than in standard
imperative languages. Declaring a label $l$ makes it invisible from the
outside of the block (while enabling it to be used inside), working just
the same way as a local variable declaration does in a standard
imperative programming language.
A declaration  removes from the model of a
labelled statement
the dependence on the hypothetical set $g_l$ of the states
attained at \code{goto}~$l$ statements.
All the instances of
\code{goto}~$l$ statements are inside the block with the declaration at
its head, so we can take a
look to see what totality of states really do accrue 
at \code{goto}~$l$ statements;
they are recognisable in the model
because they are the outcomes of the transitions that are marked with $\G_l$.
Equating the set of such states with the
hypothesis $g_l$  gives the (least) fixpoint $g_l^*$ required  in the
\code{label}~$l$ model.
%

The hypothetical sets $g_l$ of states that obtain at \code{goto}~$l$
statements are used at the point where the label $l$ appears within the
scope of the declaration. We say that any of the states in $g_l$ may be an
outcome of passing through the label $l$, because it may have been brought in
by a \code{goto}~$l$ statement.  That is an overestimate;  in reality,
if the state just before the label is $s_1$, then at most those states
$s_2$ in $g_l$ that are reachable at a \code{goto}~$l$ from an initial
program state $s_0$ that also leads to $s_1$ (either $s_1$ first or
$s_2$ first) may obtain after the
label $l$, and that may be considerably fewer $s_2$ than we calculate
in $g_l^*$.
Here is a visualisation of such a situation; the curly arrows denote a trace:
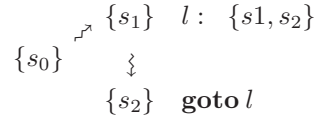
\[
\begin{array}{ccc@{\quad}l@{\quad}l}
    & & \{s_1\}&l:&\kern-20pt\{s1,s_2\}\\[-1ex]
    &\rotatebox{45}{$\leadsto$}&\\[-1ex]
\{s_0\}& &\raisebox{1.25ex}{\rotatebox{-90}{$\leadsto$}}\\[-1ex]
    &&\\[-1ex]
    & & \{s_2\}&\code{goto}~l
\end{array}
\]
If the initial precondition on the code admits more than one
initial state $s_0$ then the model may admit more states $s_2$
after the label $l$ than occur in reality when $s_1$ precedes $l$,
because the model does not take into account the dependence of $s_2$ on
$s_1$ through $s_0$. It is enough for the model that $s_2$ proceeds
from some $s_0$ and $s_1$ proceeds from some (possibly different) $s_0$
satisfying the same initial condition.
In mitigation, \code{goto}s are sparsely
distributed in real codes and we have not found the effect 
pejorative.

\begin{example}
Consider the code $R$ and suppose the input is restricted 
to a unique state $s$:
\[
\code{label}~A,B.\overbrace{
\underbrace{\code{skip};~\code{goto}~A ; ~B:~\code{return}
;~A}_Q:~\code{goto}~B}^P
\]
with labels $A$, $B$ in scope in body $P$, and the marked fragment $Q$. The 
single transitions made in the code $P$ 
and the corresponding statement sequences are:
\begin{align*}
&s\mathop{\mapsto}\limits^{\N} s \mathop{\mapsto}\limits^{{\G}_A} s
 &\#~& \code{skip};~\code{goto}~A;\\[-0.75ex]
&s \mathop{\mapsto}\limits^{\N} s \mathop{\mapsto}\limits^{\N} s
\mathop{\mapsto}\limits^{{\G}_B} s
 &\#~& \code{skip};~\code{goto}~A;A:~\code{goto}~B\\[-0.75ex]
&s \mathop{\mapsto}\limits^{\N} s \mathop{\mapsto}\limits^{\N} s \mathop{\mapsto}\limits^{{\N}} s \mathop{\mapsto}\limits^{\R} s
 &\#~& \code{skip};~\code{goto}~A;A:~\code{goto}~B;B:~\code{return}
\end{align*}
with observed states
$g_A = \{ s \}$,
$g_B = \{ s \}$ at the labels $A$ and $B$ respectively.

The $\code{goto}~B$  statement is not in the
fragment $Q$ so there is no way of knowing about the set  
of states at $\code{goto}~B$ while examining $Q$. Without that input,
the traces of $Q$ are
\begin{align*}
&s \mathop{\mapsto}\limits^{\N} s \mathop{\mapsto}\limits^{{\G}_A} s 
 &\#~&\code{skip};~\code{goto}~A\hspace{1in}\\
&s \mathop{\mapsto}\limits^{\N} s \mathop{\mapsto}\limits^{\N} s 
 &\#~&\code{skip};~\code{goto}~A;A:\hspace{1in}
\end{align*}
There are no possible entries at $B$ originating from within $Q$ itself. That is, the
model $\llbracket Q\rrbracket_g$ of $Q$ as a set of transitions
assuming $g_B = \{\,\}$,
meaning there are no entries
from outside, is $\llbracket Q \rrbracket_g  = \{
s\mathop{\mapsto}\limits^{\N}s,s\mathop{\mapsto}\limits^{\G_A}s 
\}$.

When we hypothesise $g_B = \{ s \}$ for
$Q$, then $Q$ has more traces:
\begin{align*}
&s \mathop{\mapsto}\limits^{\N} s \mathop{\mapsto}\limits^{\N} s \mathop{\mapsto}\limits^{{\N}} s \mathop{\mapsto}\limits^{\R} s
 &\#~& \code{skip};~\code{goto}~A;A:~\code{goto}~B;B:~\code{return}
\end{align*}
corresponding to these entries at $B$ from the rest of the code
proceeding to the \code{return} in $Q$, 
and $\llbracket Q\rrbracket_g = \{ s\mathop{\mapsto}\limits^{\N}s,~
s\mathop{\mapsto}\limits^{{\G}_A}s,~ s\mathop{\mapsto}\limits^{\R}s \}$. In
the context of the whole code $P$, that is the model for $Q$ as a
set of initial to final state transitions. 
\label{ex:6}
\end{example}

\begin{example}
Staying with the code of Example~\ref{ex:6},
the set $\{ s\mathop{\mapsto}\limits^{{\G}_A}s,~
s\mathop{\mapsto}\limits^{{\G}_B}s,~ s\mathop{\mapsto}\limits^{\R}s \}$
is the model $\llbracket P\rrbracket_g$ of $P$ starting at state $s$ with
assumptions $g_A$, $g_B$ of Example~\ref{ex:6},
and the sets $g_A$, $g_B$ are
observed at the labels $A$, $B$ in the code under these assumptions.
Thus $\{A\mapsto g_A, B\mapsto g_B\}$ is 
the fixpoint $g^*$ of the {\bf label} declaration rule in
Table~\ref{tab:interpretation}.

That rule says to next remove transitions ending at
\code{goto}~$A$s  and $B$s from visibility in the model of the declaration block,
because they can go nowhere else, leaving only 
$\llbracket R\rrbracket_{\{\,\}}  = \{ s\mathop{\mapsto}\limits^{\R}s\}$
as the set-of-transitions model of the whole block of code, which
corresponds to the sequence
$\code{skip};\code{goto}~A;A:~\code{goto}~B;B:~\code{return}$.

\end{example}

\noindent
We extend the propositional language to ${\mathscr B}^*$ which includes the
modal operators ${\N}$, ${\R}$, ${\B}$, ${\G}_l$, ${\e}_k$ for $l\in L$, $k\in K$,
as shown in Table~\ref{tab:B*}, which defines
a model of $\mathscr B^*$ on transitions. The predicate $\N p$ informally
should be read as picking out from the set of all coloured state transitions
`those normal-coloured transitions that
produce a state  satisfying $p$', and similarly for the other operators.
\begin{table}[t]
\caption{%
Extending the language $\mathscr B$ of propositions to modal operators ${\N}$, ${\R}$,
${\B}$, ${\G}_l$, ${\e}_k$ for $l\in L$, $k\in K$. An evaluation
on transitions
 is given for $b\in \mathscr B$, $b^* \in \mathscr B^*$.
}
\footnotesize
\label{tab:B*}
\[
\mathscr B^*~{:}{:}{\text{\---}}~
   b \mid {\N}\,b^* \mid {\R}\,b^* \mid {\B}\,b^* \mid {\G}_l\,b^* \mid {\e}_k\,b^*
   \mid b^* \lor b^*
   \mid b^* \land b^*
   \mid \lnot b^*
\]
\[
\begin{array}[t]{r@{~}c@{~}l}
\llbracket b
\rrbracket(s_0\mathop{\mapsto}\limits^{\iota} s_1)&=&\llbracket b\rrbracket s_1
\\
\llbracket {\N}\,b^* \rrbracket(s_0\mathop{\mapsto}\limits^{\iota} s_1) &=&
(\iota={\N}) \land \llbracket b^*\rrbracket (s_0\mathop{\mapsto}\limits^{\iota} s_1)
\\
\llbracket {\R}\,b^* \rrbracket(s_0\mathop{\mapsto}\limits^{\iota} s_1) &=&
(\iota={\R}) \land \llbracket b^*\rrbracket (s_0\mathop{\mapsto}\limits^{\iota} s_1)
\\
\llbracket {\B}\,b^* \rrbracket(s_0\mathop{\mapsto}\limits^{\iota} s_1) &=&
(\iota={\B}) \land \llbracket b^*\rrbracket (s_0\mathop{\mapsto}\limits^{\iota} s_1)
\\
\llbracket {\G}_l\,b^* \rrbracket(s_0\mathop{\mapsto}\limits^{\iota} s_1) &=&
(\iota={\G}_l) \land \llbracket b^*\rrbracket (s_0\mathop{\mapsto}\limits^{\iota} s_1)
\\
\llbracket {\e}_k\,b^* \rrbracket(s_0\mathop{\mapsto}\limits^{\iota} s_1) &=&
(\iota={\e}_k) \land \llbracket b^*\rrbracket (s_0\mathop{\mapsto}\limits^{\iota} s_1)
\end{array}
\vspace{-2ex}
\]
\end{table}
The modal operators 
satisfy the algebraic laws given in Table~\ref{tab:modal}.
Additionally, however, for non-modal $p\in \mathscr B$,
\begin{equation}
p = {\N} p \lor {\R} p \lor {\B} p \lor \lor\kern-6pt\lor {\G}_l p
\lor \kern-6pt\lor {\e}_k p
\label{eq:star}
\end{equation}
because each transition must be some colour, and those are all the
colours. The decomposition works in the general case too:

\begin{proposition}
Every $p\in \mathscr B^*$ can be (uniquely) expressed as
\[
p = {\N} p_{\N} \lor {\R} p_{\R} \lor {\B} p_{\B} \lor
\lor\kern-6pt\lor {\G}_l p_{\G_l}
\lor \kern-6pt\lor {\e}_k p_{\e_k}
\]
for some $p_{\N}$, $p_{\R}$, etc that are free of modal operators.
\label{prop:1}
\end{proposition}

\begin{proof}
\em
Equation  \eqref{eq:star} gives the result for $p\in \mathscr B$. The
rest is by structural induction on $p$, using 
Table~\ref{tab:modal} and boolean algebra. Uniqueness follows because
${\N}p_{\N} = {\N}p_{\N}'$, for example, applying ${\N}$ to two
possible decompositions, and applying the orthogonality and idempotence
laws;  apply the definition of
${\N}$ in the model in Table~\ref{tab:B*} to deduce $p_{\N}=
p_{\N}'$ for non-modal predicates $p_{\N}$, $p_{\N}'$. Similarly for ${\B}$, ${\R}$,
${\G}_l$, ${\e}_k$.
\hfill\EOP
\end{proof}
\begin{table}[t]
\caption{Laws of the modal operators ${\N}$, ${\R}$, ${\B}$, ${\G}_l$, ${\e}_k$
with $M,M_1,M_2\in \{{\N},{\R},{\B},{\G}_l,{\e}_k\mid l\in L,k\in K\}$ and $M_1\ne M_2$.
}
\label{tab:modal}
\footnotesize
\begin{align*}
M(\bot) &= \bot     &\text{(flatness)}
\\
M(b_1\lor b_2) &= M(b_1)\lor M(b_2) &\text{(disjunctivity)}
\\
M(b_1\land b_2) &= M(b_1)\land M(b_2) &\text{(conjunctivity)}
\\
M(M b) &= M b  &\text{(idempotence)}
\\
M_2(M_1 b) = M_1(b) \land M_2(b) &= \bot&\text{(orthogonality)}
\end{align*}
\end{table}

\noindent
So modal formulae $p\in \mathscr{B}^*$ may be viewed as tuples
$(p_{\N},p_{\R},p_{\B},p_{{\G}_l},p_{{\e}_k})$ of non-modal formulae
from $\mathscr{B}$ for labels $l\in L$, exception kinds $k\in K$.
That means that ${\N} p \lor {\R} q$, for example, is simply a
convenient notation for writing down two assertions at once: one that
asserts $p$ of the final states of the transitions that end `normally',
and one that asserts $q$ on the final states of the
transitions that end in a `return flow'. The meaning of ${\N} p \lor
{\R} q$ is  the union of the set of the normal transitions with final
state that satisfy $p$ 
plus the set of the transitions that end in a `return flow' and whose
final states satisfy $q$. We can now give meaning to a
notation that looks like (and is intended to signify) a Hoare
triple with an explicit context of certain `\code{goto} assumptions':

\begin{definition}
Let $g_l = {\llbracket p_l \rrbracket}$ be the set of states satisfying
$p_l\in \mathscr B$, labels $l\in L$.
Then `${\G}_l\,p_l \triangleright \{p\}~a~\{q\}$',
for non-modal $p,p_l\in \mathscr B$, $P\in \mathscr C$ and $q\in
\mathscr B^*$, means:
\begin{align*}
\llbracket {\G}_l\,p_l\triangleright\{p\}~P~\{q\}\rrbracket
&= \llbracket \{p\}~P~\{q\}\rrbracket_g
\\
&= \forall
 s_0\mathop{\mapsto }\limits^{\iota} s_1
   \in \llbracket P \rrbracket_g.~
   \llbracket p \rrbracket s_0 \Rightarrow
   \llbracket q \rrbracket (s_0\mathop{\mapsto }\limits^{\iota} s_1)
\end{align*}
\label{def:A5}
\end{definition}
That is read as `the triple $\{p\}~P~\{q\}$ holds under assumptions
$p_l$ at $\code{goto}~l$ when every transition of $P$ that starts at a
state satisfying $p$ also satisfies $q$'.  The explicit Gentzen-style
assumptions $p_l$ are free of modal operators.  What is meant by
the notation is that those states that may be attainable as the program
traces pass through \code{goto} statements are assumed to be restricted
to those that satisfy $p_l$.

The ${\G}_l\,p_l$ assumptions may be separated by commas, as
${\G}_{l_1}\,p_{l_1}, {\G}_{l_2}\,p_{l_2},\dots$, with $l_1\ne l_2$,
etc.  Or they may be written as a disjunction ${\G}_{l_1}\,p_{l_1}\lor
{\G}_{l_2}\,p_{l_2}\lor\dots$ because the information in this modal
formula is only the mapping $l_1\mapsto p_{l_1}$, $l_2\mapsto p_{l_2}$,
etc.  If the same $l$ appears twice among the
disjuncts ${\G}_l\,p_l$, then we understand that the union of the two
$p_l$ is intended.

Now we can prove the validity of laws about triples drawn
from what Definition~\ref{def:A5} says.  The first laws are
strengthening and weakening results on pre- and postconditions:

\begin{proposition}
The following algebraic relations hold:
\begin{align}
\llbracket \{\bot\}~P~\{q\} \rrbracket_g
  &{\iff} \top
\label{eq:G1}
\\
\llbracket \{p\}~P~\{\top\} \rrbracket_g
  &{\iff} \top
\label{eq:G2}
\\
\llbracket \{p_1\lor p_2\}~P~\{q\} \rrbracket_g
  &{\iff}
        \llbracket \{p_1\}~P~\{q\} \rrbracket_g
        \land \llbracket \{p_2\}~P~\{q\} \rrbracket_g
\label{eq:G3}
\\
\llbracket \{p\}~P~\{q_1\land q_2\} \rrbracket_g
   &{\iff}
        \llbracket \{p\}~P~\{q_1\} \rrbracket_g
        \land \llbracket \{p\}~P~\{q_2\}  \rrbracket_g
\label{eq:G4}
\\
(p_1{\to} p_2) \land\llbracket \{p_2\}~P~\{q\} \rrbracket_g
   &\implies
   \llbracket \{p_1\}~P~\{q\} \rrbracket_g
\label{eq:G5}
\\
(q_1{\to} q_2) \land\llbracket \{p\}~P~\{q_1\} \rrbracket_g
   &\implies
   \llbracket \{p\}~P~\{q_2\} \rrbracket_g
\label{eq:G6}
\\
\llbracket \{p\}~P~\{q\} \rrbracket_{g'}
   &{\implies}
\llbracket \{p\}~P~\{q\} \rrbracket_g
\label{eq:G7}
\end{align}
for $p,p_1,p_2\in \mathscr B$, $q,q_1,q_2\in \mathscr B^*$, $P\in
\mathscr C$, and $g_l \subseteq g'_l\in\mathds{P}S$.
\label{prp:P3}
\end{proposition}
\begin{proof}
\em
(\ref{eq:G1}-\ref{eq:G4}) follow on applying Definition~\ref{def:A5}.
(\ref{eq:G5}-\ref{eq:G6}) follow from (\ref{eq:G3}-\ref{eq:G4})
on considering the cases $p_1\lor p_2 = p_2$ and $q_1\land q_2 = q_1$.
The reason for \eqref{eq:G7} is that $g'_l$ is a bigger set 
than $g_l$, so $\llbracket P \rrbracket_{g'}$ is a bigger set of
transitions than $\llbracket P \rrbracket_g$
and thus the universal
quantifier in Definition~\ref{def:A5} produces a smaller (less
true) truth value.
\hfill\EOP
\end{proof}
\begin{theorem}[Soundness]
The following algebraic inequalities hold, for $\E{}_1$
any of ${\R}$, ${\B}$, ${\G}_l$, ${\e}_k$; $\E_2$ any of ${\R}$, ${\G}_l$, ${\e}_k$;
$\E_3$ any of ${\R}$, ${\B}$, ${\G}_{l'}$ for $l'\ne l$, ${\e}_k$;
$\E_4$ any of ${\R}$, ${\B}$, ${\G}_{l}$, ${\e}_{k'}$ for $k'\ne k$;
$[h]$ the code of the subroutine called $h$:

\begin{small}
\begin{align}
\left.\begin{array}{@{}l@{~}l}
&\llbracket \{p\}\, P\, \{{\N}q\lor \E_{1}x\}\rrbracket_g\\
\land&
\llbracket \{q\}\, Q\, \{{\N}r\lor \E_{1}x\}\rrbracket_g
\end{array}\right\}
&{\implies}
\llbracket \{p\}\, P\,{;}\,Q\,\, \{{\N}r\lor \E_{1}x\}\rrbracket_g
\label{eq:AT9}
\\
\llbracket \{p\}\, P\, \{\B q\lor {\N}p\lor \E_{2}x\}\rrbracket_g
&{\implies}
\llbracket \{p\}\, \code{do} \,\,P\, \{{\N}q\lor \E_{2}x\}\rrbracket_g
\\
\top
&{\implies} \llbracket \{p\}\, \code{skip}\, \{{\N}\,p\}\rrbracket_g
\\
\top
&{\implies} \llbracket \{p\}\, \code{return}\, \{{\R}\,p\}\rrbracket_g
\\
\top
&{\implies} \llbracket \{p\}\, \code{break}\,\,\, \{{\B}\,p\}\rrbracket_g
\\
\top
&{\implies} \llbracket \{p\}\, \code{goto}\,\,l\, \{{\G}_l\,p\}\rrbracket_g
\label{eq:AT14}
\\
\top
&{\implies} \llbracket \{p\}\, \code{throw}\,\,k\, \{{\e}_k\,p\}\rrbracket_g
\label{eq:AT15}
\\
\llbracket \{b\land p\}\, P\, \{q\}\rrbracket_g
&{\implies}
\llbracket \{p\}\, b\,{\to}P\,\, \{q\}\rrbracket_g
\\
\llbracket \{p \}\, P\,\, \{q\}\rrbracket_g
\land
\llbracket \{p \}\, Q\,\, \{q\}\rrbracket_g
&{\implies} 
\llbracket \{p \}\, P\,{\shortmid}\,Q\,\, \{q\}\rrbracket_g
\\
\top
&{\implies} \llbracket \{q[e/x] \}\,\, x{=}e\,\, \{{\N}q\}\rrbracket_g
\\
\llbracket \{p \}~ P~ \{q\}\rrbracket_g
\land
g_l\subseteq \{s_1\mid s_0\mathop{\mapsto}\limits^{\N}s_1\in\llbracket q\rrbracket
\}
&{\implies} \llbracket \{p \}~ P:l~ \{q\}\rrbracket_g
\label{eq:AT18}
\\
\llbracket \{p \}~ P~ \{{\G}_l p_l \lor {\N}q \lor \E_{3}x\}\rrbracket_{g\cup\{l\mapsto p_l\}}
&{\implies} 
\llbracket \{p \}~ \code{label}~l.P~ \{{\N}q\lor \E_{3}x\}\rrbracket_g
\\
\llbracket \{p \}~ [h]~ \{\R r \lor {\e}_k x_k\}\rrbracket_{\{~\}}
&{\implies} 
\llbracket \{p \}~ \code{call}~h~ \{{\N}r\lor {\e}_k x_k\}\rrbracket_{g}
\label{eq:AT20}
\\
\left.\begin{array}{@{}l@{~}l}
&\llbracket \{p \}~ P~ \{{\N}r\lor  {\e}_k q \lor\E_{4} x \}\rrbracket_{g}\\~
\land&
\llbracket \{q \}~ Q~ \{{\N}r\lor {\e}_k x_k\lor \E_{4} x  \}\rrbracket_{g}
\end{array}\right\}
&{\implies} 
\llbracket \{p \}~ \code{try}~P~\code{catch}(k)~Q~ \{{\N}r\lor {\e}_k x_k\lor \E_{4} x \}\rrbracket_{g}
\label{eq:AT21}
\end{align}
\end{small}
\label{thm:T1}
\end{theorem}
\begin{proof}
By evaluation, given Definition~\ref{def:A5} and the semantics from
Table~\ref{tab:interpretation}.
\hfill\EOP
\end{proof}
The reason why the theorem is titled `Soundness' is that its
inequalities can be read as the NRB logic deduction rules set out in
Table~\ref{tab:NRBG}, via Definition~\ref{def:A5}.
The fixpoint requirement of the model at the
\code{label} construct is expressed in the `arrival from a \code{goto} at a
label' law \eqref{eq:AT18}, where it is stated that {\em if} the hypothesised
states $g_l$ at a \code{goto}~$l$ statement are covered by the states
$q$ immediately after code block $P$ and preceding label $l$,
{\em then} $q$ holds
after the label $l$ too. However, there is no need for any such
predication when the $g_l$ are exactly the fixpoint of the map
\[ g_l\mapsto \{s_1\mid s_0\mathop{\mapsto}\limits^{\G_l}
s_1\in\llbracket P\rrbracket_g\}
\]
because that is what the fixpoint condition says.  Thus, while the model
in Table~\ref{tab:interpretation} satisfies equations
(\ref{eq:AT9}-\ref{eq:AT21}), it satisfies more than they require --
some of the hypotheses in the equations could be dropped and the model
would still satisfy them. But the NRB logic rules in
Table~\ref{tab:NRBG} are validated by the model and thus are sound.

\section{Completeness for deterministic programs}
\label{sec:A2}

In proving completeness of the NRB logic, at least for deterministic
programs,  we will be guided by the proof of partial completeness for
Hoare's logic in
K.~R.~Apt's survey paper \cite{KRAPT}.  We will need, for
every (possibly modal) postcondition $q \in \mathscr B^*$ and every
construct $R$ of $\mathscr C$, a non-modal formula $p\in
\mathscr B$ that is weakest in $\mathscr B$ such that if $p$ holds of a
state $s$, and $s\mathop{\mapsto}\limits^\iota s'$ is in the model of
$R$ given in Table~\ref{tab:interpretation}, then $q$ holds of
$s\mathop{\mapsto}\limits^\iota s'$.  This $p$ is written
$\mbox{wp}(R,q)$, the `weakest precondition on $R$ for $q$'.  We
construct it via structural induction on $\mathscr C$ at the same
time as we deduce completeness, so
there is an element of chicken versus egg about the proof, and we will
not labour that point.

We will also suppose that we can prove any tautology of $\mathscr B$
and $\mathscr B^*$, so `completeness of NRB' will be relative to that
lower-level completeness.

Notice that there is always a set $p\in \mathds{P}S$ satisfying
the `weakest precondition' characterisation above.  It is $\{ s\in S\mid
s\mathop{\mapsto}\limits^\iota s'\in\llbracket R\rrbracket_g \Rightarrow
s\mathop{\mapsto}\limits^\iota s' \in \llbracket q\rrbracket \}$, and
it is called the weakest {\em semantic} precondition on $R$ for $q$.
So we sometimes refer to $\text{wp}(R,q)$  as the `weakest {\em syntactic}
precondition' on $R$ for $q$, when we wish to emphasise the distinction.
The question is whether or not there is a formula in $\mathscr B$ that
exactly expresses this set.  If there is, then the system is said to be
{\em expressive}, and that formula {\em is} the weakest (syntactic)
precondition on $R$ for $q$, $\text{wp}(R,q)$.  Notice also that a
weakest (syntactic) precondition $\text{wp}(R,q)$ must encompass the
semantic weakest precondition; that is because if there were a state $s$
in the latter and not in the former, then we could form the disjunction
$\text{wp}(R,q)\lor(x_1=s x_1\land\dots x_n=s x_n)$ where the $x_i$ are
the variables of $s$, and this would also be a precondition on $R$ for
$q$, hence $x_1=s x_1\land\dots x_n=s x_n \to \text{wp}(R,q)$ must be
true, as the latter is supposedly the weakest precondition, and so $s$
satisfies $\text{wp}(R,q)$ in contradiction to the assumption that $s$
is not in $\text{wp}(R,q)$. For orientation, then, the reader should note
that `there is a weakest (syntactic) precondition in
$\mathscr B$' means there is a unique strongest formula in
$\mathscr B$ covering the weakest semantic precondition.

We will lay out the proof of completeness inline here, in order to
avoid excessively overbearing formality, and at the end we will draw the
formal conclusion. 

A completeness proof is always a proof by cases on each construct of
interest.  It has the form `suppose that {\em foo} is true, then we can
prove it like this', where \emph{foo} runs through all the constructs we
are interested in.  We start with assertions
about the sequence construction $P;Q$.  We will look at this in
particular detail, noting where and how the weakest precondition formula
plays a role, and skip that detail for most other cases.
Thus we start with \emph{foo} equal to $\G_l\,
g_l~\triangleright~\{p\}~P;Q~\{q\}$ for some assumptions $g_l\in
\mathscr{B}$, but we do not need to take the assumptions $g_l$ into
account in this case.

\medskip

\begin{em}
\paragraph{Case $P;Q$}.
Consider a sequence of two statements $P;Q$ for which
$\{p\}~P;Q~\{q\}$ holds in the model set out by Definition~\ref{def:A5}
and Table~\ref{tab:interpretation}.  That is, suppose that initially the
state $s$ satisfies predicate $p$ and that there is a progression from
$s$ to some final state $s'$ through $P;Q$.  Then
$s\mathop{\mapsto}\limits^\iota s'$ is in $\llbracket P;Q\rrbracket_g$
and $s\mathop{\mapsto}\limits^\iota s'$ satisfies $q$.  We will consider
two subcases, the first where $P$ terminates normally from $s$, and the
second where $P$ terminates abnormally from $s$.  A third possibility,
that $P$ does not terminate at all, is ruled out because a final state
$s'$ is reached.

Consider the first subcase, which means that we think of $s$
as confined to $\mbox{wp}(P,\N\top)$.
According to 
Table~\ref{tab:interpretation}, that means that $P$ started in state
$s_0=s$ and finished normally in some state $s_1$
and $Q$ ran on from state $s_1$ to finish normally in state $s_2=s'$.
Let $r$ stand for the weakest precondition $\mbox{wp}(Q,\N q)$ that
guarantees a normal termination of $Q$ with $q$ holding.
By definition of weakest precondition, $\{r\}~Q~\{\N q\}$,
is true and $s_1$ satisfies $r$ (if not, then $r\lor (x_1=s x_1\land
x_2= s x_2 \land \dots)$ would be a weaker precondition for $\N q$ than
$r$, which is impossible). The latter is true whatever
$s_0$ satisfying $p$ and $\mbox{wp}(P,\N\top)$ we started with, so
by definition of weakest precondition, $p\land \mbox{wp}(P,\N\top)\to
\mbox{wp}(P,\N r)$ must be true, which is to say that $\{p\land
\mbox{wp}(P,\N\top)\}~P~\{\N r\}$ is true.

By induction, it is the case that there are deductions
$\vdash \{p\land \mbox{wp}(P,\N\top)\}~P~\{\N r\}$ and $\vdash
\{r\}~Q~\{\N q\}$ in the NRB system.
But the following rule
\[
\frac{
               \{p\land \mbox{wp}(P,\N\top)\}~P~\{\N r\} \quad  \{r\}~Q~\{\N q\}
}{
               \{p\land \mbox{wp}(P,\N\top)\}~P;Q~\{\N q\}
}
\]
is a derived rule of NRB logic. It is a specialised form of the
general NRB rule of sequence.  Putting these deductions together,
we have a deduction of the truth of the assertions
$ \{p\land \mbox{wp}(P,\N\top)\}~P;Q~\{\N q\}$.
By weakening on the conclusion, since $\N q\to q$ is (always) true,
we have a deduction of $ \{p\land \mbox{wp}(P,\N\top)\}~P;Q~\{q\}$.

Now consider the second subcase, when the final state $s_1$ reached from
$s=s_0$
through $P$ obtains via an abnormal flow out of $P$. This means
that we think of $s$ as confined to $\mbox{wp}(P,\lnot\N\top)$.
Now the transition $s_0\mathop{\mapsto}\limits^\iota s_1$ in 
$\llbracket P\rrbracket_g$ satisfies $q$, and $s$ is arbitrary in 
$p\land\mbox{wp}(P,\lnot\N\top)$, so
$\{p\land\mbox{wp}(P,\lnot\N\top)\}~P~\{q\}$.
However, `not ending normally' (and getting to a termination, which is
the case here) means `ending abnormally', i.e.,
$\R\top\lor\B\top\lor\dots$ through all of the available colours, as per
Proposition~\ref{prop:1}, and we may write the assertion out as
$\{p\land\mbox{wp}(P,\R\top\lor\B\top\dots)\}~P~\{q\}$. 
Considering the cases separately, one has
$\{p\land\mbox{wp}(P,\R\top)\}~P~\{\R q\}$ (since $\R q$ is the
component of $q$ that expects an $\R$-coloured transition), and
$\{p\land\mbox{wp}(P,\B\top)\}~P~\{\B q\}$, and so on, all holding.
By induction, there are deductions
$\vdash \{p\land\mbox{wp}(P,\R\top)\}~P~\{\R q\}$,
$\vdash \{p\land\mbox{wp}(P,\B\top)\}~P~\{\B q\}$,
etc.  But the following rule
\[
\frac{
 \{p \land\mbox{wp}(P,\E\top)\}~P~\{\E q\}
}{
\{p \land\mbox{wp}(P,\E\top)\}~P;Q~\{\E q\}
}
\]
is a derived rule of NRB logic for each `abnormal' colouring $\E$, and
hence we have a deduction $\vdash \{p
\land\mbox{wp}(P,\E\top)\}~P;Q~\{\E q\}$ for each of the `abnormal'
colours $\E$.  By weakening on the conclusion, since $\E q\to q$, for
each of the colours $\E$, we have a deduction $\vdash \{p
\land\mbox{wp}(P,\E\top)\}~P;Q~\{q\}$ for each of the colours $\E$.

By the rule on disjunctive hypotheses (fourth
from last in Table~\ref{tab:NRBG}) we now have a deduction $\vdash \{p\land
(\mbox{wp}(P,\N\top)\lor\mbox{wp}(P,\R\top)\lor\dots)\}~P;Q~\{q\}$.
But the weakest precondition is monotonic, so 
$\mbox{wp}(P,\N\top)\lor\mbox{wp}(P,\R\top)\lor\dots$ is covered by
$\mbox{wp}(P,\N\top\lor\R\top\lor\dots)$, which is $\mbox{wp}(P,\top)$
by  Proposition~\ref{prop:1}. But for a deterministic program $P$, the
outcome from a single starting state $s$ can only be uniquely a normal
termination, or uniquely a return termination, etc, and
$\mbox{wp}(P,\N\top)\lor\mbox{wp}(P,\R\top)\lor\dots =
\mbox{wp}(P,\N\top\lor\R\top\lor\dots) = \mbox{wp}(P,\top)$ exactly.
The latter is just $\top$, so we have
a proof $\vdash \{p\}P;Q~\{q\}$.
As to what the weakest precondition $\mbox{wp}(P;Q,q)$ is, it is
$\mbox{wp}(P,\N \mbox{wp}(Q,q))\lor
\mbox{wp}(P,\R q)\lor \mbox{wp}(P, \B q)\lor \dots$, the disjunction
being over all the possible colours.
\end{em}

\medskip
That concludes the consideration of the case $P;Q$.
The existence of a
formula expressing a weakest precondition is what really drives the
proof above along, and in lieu of pursuing the proof through all the
other construct cases,  we note the important weakest precondition formulae
below:
\begin{itemize}
\item
The weakest precondition for assignment is
$\mbox{wp}(x=e,\N q) = q[e/x]$ for $q$ without modal components. In
general $\mbox{wp}(x=e, q) = \N q[e/x]$.
\item 
The weakest precondition for a \code{return} statement is
$\mbox{wp}(\code{return},q) = \R q$.
\item 
The weakest precondition for a \code{break} statement is
$\mbox{wp}(\code{break},q) = \B q$. Etc.
\item
The weakest precondition $\mbox{wp}(\code{do}~P,\N q)$ for a \code{do} loop
that ends `normally' is 
$\code{wp}(P,\B q) \lor \code{wp}(P,\N \code{wp}(P,\B q)) \lor
\code{wp}(P,\N \code{wp}(P,\N \code{wp}(P,\B q))) \lor \dots$.
That is, we might break from $P$ with $q$, or run through $P$
normally to the precondition for breaking from $P$ with $q$ next, etc.
Write $\code{wp}(P,\B q)$ as $p$ and write $\code{wp}(P,\N r)\land
\lnot p$ as $\psi(r)$, Then $\mbox{wp}(\code{do}~P,\N q)$ can be
written $p \lor \psi(p) \lor \psi(p\lor \psi(p)) \lor \dots$, which is
the strongest solution to $\pi = \psi(\pi)$ no stronger than $p$.
This is the weakest precondition for $p$ after $\code{while} (\lnot p)~
P$ in classical Hoare logic. 
It is an existentially quantified statement, stating that an initial
state $s$ gives rise to exactly some $n$ passes through $P$ before the
condition $p$ becomes true for the first time.  It can classically
be expressed as a formula of first-order logic and it is the weakest
precondition for $\N q$ after $\code{do}~P$ here.

The preconditions for $\E q$ for each `abnormal' coloured ending $\E$ of
the loop $\code{do}~P$ are similarly expressible in $\mathscr B$, and
the precondition for $q$ is the disjunction of each of the preconditions for
$\N q$, $\R q$, $\B q$, etc.
\item
The weakest precondition for a guarded statement $\mbox{wp}(p\to
P,q)$ is $p\to \mbox{wp}(P,q)$, as in Hoare logic; and
the weakest precondition for a disjunction $\mbox{wp}(P\shortmid Q, q)$ is
$\mbox{wp}(P,q) \land \mbox{wp}(Q,q)$, as in Hoare logic. However, we
only use the deterministic combination $p\to P\shortmid \lnot p\to Q$
for which the weakest precondition is
$(p\to\mbox{wp}(P,q))\land(\lnot p\to\mbox{wp}(Q,q))$, i.e.
$p\land\mbox{wp}(P,q)\lor \lnot p\land\mbox{wp}(Q,q)$.
\end{itemize}
To deal with  labels properly, we have to extend some of these notions
and notations to take account of the assumptions $\G_l g_l$ that
an assertion $\G_l g_l~\triangleright~\{ p \}~P~\{ q \}$ is made against.
The weakest precondition $p$ on $P$ for $q$ is then $p = \mbox{wp}_g(P,q)$,
with the $g_l$ as extra parameters. The weakest precondition for a
label use $\mbox{wp}_g(P :l,q)$ is then 
$\mbox{wp}_g(P,q)$, provided that $g_l\to q$,
since the states $g_l$
attained by $\code{goto}~l$ statements throughout the code are
available after the label, as well as those obtained through $P$.
The weakest precondition in the general situation where it is not
necessarily the case that $g_l\to q$ holds is
$\mbox{wp}_g(P,q\land (g_l\to q))$, which is 
$\mbox{wp}_g(P,q)$.

Now we can continue the completeness proof through the statements of the
form $P:l$ (a labelled statement) and $\code{label}~l.P$ (a label
declaration).

\medskip

\begin{em}
\paragraph{Case labelled statement}.
If $\llbracket \{p\}~P:l~\{q\}\rrbracket_g$ holds,
then every state $s = s_0$ satisfying $p$ leads through $P$ with
$s_0\mathop{\mapsto}\limits^\iota s_1$ satisfying $q$, and also $q$
must contain all the transitions 
$s_0\mathop{\mapsto}\limits^{\N} s_1$ where $s_1$ satisfies
$g_l$. Thus $s$ satisfies $\mbox{wp}_g(P,q)$ and $\N g_l\to q$ holds.
Since $s$ is arbitrary in
$p$, so $p\to \mbox{wp}_g(P,q)$ holds and by induction, $\vdash \G_l g_l~\triangleright~\{p\} ~P~\{q\}$.
Then, by the `frm' rule of NRB (Table~\ref{tab:NRBG}), we may deduce
$\vdash  \G_l g_l~\triangleright~\{p\} ~P:l~\{q\}$.
\end{em}


\medskip

\begin{em}
\paragraph{Case label declaration}.
The weakest precondition for a declaration
$\mbox{wp}_g(\code{label}\,l.P,q)$ is simply $p = \mbox{wp}_{g'}(P,q)$,
where the assumptions after the declaration are
$g' = g \cup \{l\mapsto g_l\}$
and $g_l$ is such that
$\G_l g_l \triangleright \{ p \}~P~\{q\}$.
In other words, $p$ and $g_l$ are
simultaneously chosen to make
the assertion hold, $p$ maximal and $g_l$ the least fixpoint describing
the states at $\code{goto}~l$
statements in the code $P$, given that the initial state satisfies $p$
and assumptions $\G_l g_l$ hold.
The $g_l$y are the statements that after exactly some $n\in\mathds{N}$
more traversals through $P$ via $\code{goto}~l$, the trace from state $s$
will avoid another $\code{goto}~l$ for the first time and exit $P$
normally or via an abnormal exit that is not a $\code{goto}~l$.

If it is the case that $\llbracket \{p\}~\code{label}~l.P~\{q\}\rrbracket_g$
holds then every state $s=s_0$ satisfying  $p$ leads through
$\code{label}~l.P$ with
$s_0\mathop{\mapsto}\limits^\iota s_1$ satisfying $q$. That means that
$s_0\mathop{\mapsto}\limits^\iota s_1$ leads through $P$, but it is not
all that do; there are extra transitions with $\iota = \G_l$ that are
not considered.  The `missing' transitions
are precisely the $\G_l g_l$ where $g_l$ is the appropriate least
fixpoint for $g_l = \{s_1\mid s_0\mathop{\mapsto}\limits^{\G_l}s_1
\in \llbracket P\rrbracket_{g\cup\{l\mapsto g_l\}}$, which is a
predicate expressing the idea that $s_1$ at a $\code{goto}~l$ 
initiates some exactly $n$ traversals back through $P$ again before
exiting $P$ for a first time other than via a $\code{goto}~l$.
The predicate $q$ cannot mention $\G_l$ since the label $l$
is out of scope for it, but it may permit some, all or no
$\G_l$-coloured transitions. The predicate $q\lor \G_l g_l$, on the
other hand, permits all the $\G_l$-coloured transitions that exit $P$.
transitions. Thus adding $\G_l g_l$ to the assumptions means
that $s_0$ traverses $P$ via $s_0\mathop{\mapsto}\limits^\iota s_1$
satisfying $q\lor \G_l g_l$ even though more transitions are admitted.
Since $s=s_0$ is arbitrary in
$p$, so $p\to \mbox{wp}_{g\cup\{l\mapsto g_l\}}(P,q\lor \G_l g_l)$
and by induction
$\vdash \G_l~\triangleright~\{p\}~P~\{q\lor \G_l g_l\}$, and then one
may deduce
$\vdash \{p\}~\code{label}~l.P~\{q\}$ by  the `lbl' rule.
\end{em}
\hfill\EOP

\medskip

That concludes the text that would appear in a proof,
but which we have abridged and presented as a discussion here! We have
covered the typical case ($P;Q$) and the unusual cases ($P:l$,
$\code{label}~l.P$).
The proof-theoretic content of the discussion is:
\begin{theorem}[Completeness]
The system of NRB logic in Table~\ref{tab:NRBG} is complete for
deterministic programs, relative to the completeness of first-order logic.
\end{theorem}
We do not know if the result holds for non-deterministic programs too,
but it seems probable. A different proof technique would be needed 
(likely showing that attempting to construct a proof backwards
either succeeds or yields a counter-model).

Along with that we note
\begin{theorem}[Expressiveness]
The weakest precondition $\mbox{wp}(P,q)$ for $q\in \mathscr B^*$,
$P\in \mathscr C$ in the interpretation set out in
Definition~\ref{def:A5} and Table~\ref{tab:interpretation}
is expressible in $\mathscr B$.
\end{theorem}
The observation
above is that there is a formula in $\mathscr B$ that expresses the
semantic weakest precondition exactly.

\section{Summary}

We have proven the NRB logic sound with respect to a simple
transition-based model of programs, and showed that it is complete for
deterministic programs.

\noindent
\bibliographystyle{plain}
\bibliography{NRB}

\end{document}